\documentclass{article}
\usepackage{spconf}
\usepackage{graphicx}
\usepackage{amsmath,amsfonts,amscd,amssymb,amsthm}
\usepackage{bm}
\usepackage[utf8]{inputenc}
\usepackage{rotating}
\usepackage{subfigure}
\usepackage{dsfont}
\usepackage{xcolor}

\graphicspath{{images}}

\def\Dd{{\mathcal{D}}}
\def\bX{{\bm{X}}}
\def\bx{{\bm{x}}}
\def\ba{{\bm{a}}}
\def\btheta{{\bm{\theta}}}
\def\param{{\boldsymbol{\rho}}}
\def\transfonop{{\mathcal{T}}}
\def\transfo{{\transfonop_\param}}
\def\transforad{{\transfonop_{\alpha,\beta}}}

\DeclareMathOperator{\argmax}{argmax}

\newcommand{\uargmax}[1]{\underset{#1}{\argmax}\;}

\newtheorem{prop}{Proposition}

\title{Template matching with noisy patches: A contrast-invariant GLR test}
\threeauthors
{Charles-Alban Deledalle}
{IMB\\CNRS-Univ. Bordeaux 1\\France }
{Loïc Denis}
{Laboratoire Hubert Curien\\CNRS-Univ. Saint-Etienne\\France}
{Florence Tupin}
{Télécom ParisTech\\CNRS-LTCI\\France}
\begin{document}
\sloppy
\maketitle
\begin{abstract}
Matching patches from a noisy image to atoms in a dictionary of patches
is a key ingredient to many techniques in image
processing and computer vision. By representing with a single atom all
patches that are identical up to a radiometric transformation,
dictionary size can be kept small, thereby retaining good computational efficiency.
Identification of the atom in best match with a given noisy patch then
requires a contrast-invariant criterion.
%To limit the size of dictionaries,
%atoms should represent a class of patches that are identical up
%to a radiometric transformation. A dictionary should so be
%manipulated with a contrast invariant template matching criterion.
%A fundamental difficulty when matching a noisy patch to a template
%is to decide whether the differences should be ascribed to
%noise or intrinsic dissimilarity up to a radiometric transformation.
In the light of detection theory, we propose a new criterion
that ensures contrast invariance and robustness to noise.
We discuss its theoretical grounding and
assess its performance under
Gaussian, gamma and Poisson noises.
\end{abstract}
\section{Introduction}
\label{sec:intro}

In this paper, we address the problem of template matching of patches under various noise conditions.
More precisely, when provided a collection of noise-free templates (the dictionary),
we focus on finding for a given noisy patch the best matching element in the dictionary.
Template matching is at the heart of many recent image processing and computer vision techniques, for instance, for denoising \cite{elad2006image}
or classification with a labeled dictionary \cite{mairal2012task}.
We focus in the following on how to perform template matching when the
noise departs from the Gaussian distribution.
Inspired by our previous work about the comparison of noisy patches
\cite{deledalle2012compare}, we extend here the proposed methodology to the
problem of template matching.

By $\bx$, we denote a patch of an image, i.e., a collection of $N$
noisy pixel values. By $\ba \in \Dd$, we denote a template taken from a dictionary
$\Dd$ ($\ba$ also has $N$ pixels). We do not specify here a shape
but consider that the values
are ordered so that when a patch $\bx$ is compared to a template
$\ba$, values with identical index are in spatial correspondence.
For best efficiency, dictionaries should be as small as possible while being representative of images.
To limit the size of dictionaries,
a common idea is to let atoms represent a class of patches that are identical up to a radiometric transformation.
Hence, a template should essentially encode the geometrical patterns of a patch rather than its radiometry.
Of course, to exploit such a dictionary, the template matching criterion
must be invariant to the radiometric changes considered 
while being robust to the noise statistic.

We assume that the noise can be modeled by a (known) distribution
so that a noisy patch $\bx$ is a realization of an
$N$-dimensional random variable $\bX$ modeled by a probability  density or mass function
$p(\cdot|\btheta)$.
The vector of parameters $\btheta$ is referred in the
following as the noise-free patch.
%We will consider in our experiments white noise, i.e.,
%Under spatial independence assumption,
%the probability density function (pdf) of $\bX$ is
%$p(\bx | \btheta)
%= \prod_{k = 1}^N p(x_k | \theta_k)$, even if the definitions of
%all criteria are general enough to deal with correlated noise.
For example, a patch $\boldsymbol x$ damaged by additive white Gaussian
noise with standard deviation $\sigma$ can be modeled by:
\begin{align}
  \label{eq:additive_noise}
  \boldsymbol x = \boldsymbol \theta + \sigma \boldsymbol n
\end{align}
where $\boldsymbol \theta$ is the noise-free patch and $\boldsymbol n$ is
the realization of a zero-mean normalized Gaussian random vector with
independent elements.
It is straightforward to see that $\bX|\btheta$ follows
a Gaussian distribution with mean $\boldsymbol \theta$ and
standard deviation $\sigma$.
While such decompositions exist for some specific distributions
(e.g., gamma distribution involves a multiplicative
decomposition), in most cases no decomposition of 
$\boldsymbol x$ in terms of $\btheta$ and
an independent noise component may be found (e.g., under Poisson noise).
In general, when noise departs from additive Gaussian noise,
the link between $\bX$ and $\btheta$ is
described by the probability density or mass function $p(\bx|\btheta)$.

\section{Problem definition}

A template matching criterion $c$
defines a mapping from a pair formed by a noisy patch and a template
$(\bx,\ba)$ to a real value.
The larger the value of $c(\bx,\ba)$,
the more relevant the match between $\bx$ and the template $\ba$.
We consider that a matching criterion $c$
is invariant with respect to the family of transformations $\transfo$
parametrized by vector $\param$, if
\begin{align*}
  \forall \bX, \ba, \param, \quad&
  c({\bm X}, \transfo({\bm a})) = c({\bm X}, {\bm a})\,.
\end{align*}
A typical example is to consider invariance up to an affine change of contrast:
$\transfo(\ba) = \transforad(\ba) = \alpha \ba + \beta \mathds{1}$,
where $\mathds{1}_k = 1$ for all $1 \leq k \leq N$.
In the light of detection theory, we consider that a noisy patch $\bx$
and a template $\ba$ are in match (up to a transformation $\transfo$)
when $\bx$ is a realization of a random variable $\bX$
following a distribution $p(. | \btheta)$ for which there exists a
vector of parameters $\param$ such that
$\btheta = \transfo(\ba)$.
The template matching problem can then be rephrased as the following
hypothesis test (a parameter test):
\begin{equation*}
  \begin{array}{llr}
    \mathcal{H}_0 : \exists \param \quad &\boldsymbol\theta = \transfo(\ba)
    & \text{(null hypothesis)} ,\\
    \mathcal{H}_1 :
    \forall \param \quad
    & \boldsymbol\theta \ne \transfo(\ba)
    & \text{(alternative hypothesis)} .
  \end{array}
\end{equation*}

For a given template matching criterion $c$,
the \textit{probability of false alarm}
(to decide $\mathcal{H}_1$ under $\mathcal{H}_0$)
and the \textit{probability of detection}
(to decide $\mathcal{H}_1$ under $\mathcal{H}_1$)
are defined as:
\begin{align}
P_{FA}%^c(\tau)
&=
\mathbb{P}(c(\bX, \ba) < \tau |
\param, \mathcal{H}_0),\\
P_{D}%^c(\tau)
&=
\mathbb{P}(c(\bX, \ba) < \tau |
\btheta, \mathcal{H}_1).
\end{align}
Note that the inequality symbols are reversed compared to usual definitions
since we consider detection of mismatch based on the matching measure $c$.

%Let $\bx$ and $\bx_2$
%be two realizations
%of the respective random variable
%$\bX$ and $\bX_2$.
%i.e. $x = X(F)$ and $x_2 = X_2(F_2)$ where $F$ and $F_2$
%are two arbitrary scene configurations in $\mathcal{F}$.
According to Neyman-Pearson theorem,
the optimal criterion, i.e., the criterion which maximizes $P_D$
for any given $P_{FA}$, is the likelihood ratio (LR) test:
% Rq: c'est peut-être dommage de donner deux noms: un sigle (LR)
% et un symbole (L). Qu'est-ce que tu en penses?
\begin{align}
  \label{eq:lr}
  \mathcal{L}(\bx, \ba)
  =
  \frac{p(\bx |
    \btheta = \transfo(\ba),
    \mathcal{H}_0)
  }{p(\bx |
    \btheta,
    \mathcal{H}_1)
  } .
\end{align}
The application of the likelihood ratio test requires the knowledge of
$\param$ and $\btheta$
(the parameters of the transformation and the noise-free patch) which, of course, are
unavailable.
%Unfortunately, it is clear that eq.~(\ref{eq:lr}) cannot be evaluated
%since parameters $\btheta$,
%$\btheta$ and $\btheta_{12}$
%are unknown.
%Such a problem is called a \textit{composite hypothesis problem}.
Our problem is thus a \textit{composite hypothesis problem}.
A criterion maximizing $P_D$ for all $P_{FA}$ and all
values of the unknown parameters is said \textit{uniformly most powerful}
(UMP).
Kendall and Stuart (1979) showed that no
UMP detector exists in general for our \textit{composite hypothesis problem}
\cite{kendall1979advanced},
so that any criteria can be defeated by another one
at a specific $P_{FA}$.
The research of a universal template matching criterion is then futile.
We address here the question of how different
criteria behave on patches extracted from natural images.

\section{Contrast-invariant template matching}
\label{sec:patchsim}

In this section we consider radiometric changes $\transforad$ defined by
two parameters: $\alpha$ and $\beta$. We present different candidate criteria for
contrast-invariant template matching and discuss their robustness to
the noise statistics.

{\bfseries{Normalized correlation:}}
The most usual way to measure similarity up to an affine change of contrast of the form
$\transforad(\bx) = \alpha \bx + \beta \mathds{1}$
between two (non-constant) vectors $\bx$ and $\ba$
is to consider their normalized correlation:
\begin{align}
  \label{eq:gauss_kernel}
  \mathcal{C}(\bx, \ba)
  =
  \left|
    \frac{
      \sum_{k=1}^N (x_k - \bar{x}) (a_k - \bar{a})
    }{
      \sqrt{\sum_{k=1}^N (x_k - \bar{x})^2 \sum_{k=1}^N (a_k - \bar{a})^2}
    }
  \right|~.
\end{align}
where $\bar{x} = \frac{1}{N} \sum_k x_k$ and
$\bar{a} = \frac{1}{N} \sum_k a_k$.
Indeed, it is straightforward to show that
the correlation provides the desired contrast invariance property.
Regarding noise corruptions, it is not straightforward whether the
correlation is a robust template matching criterion.
We will show that, under the assumption of Gaussian noise,
for a fixed observation $\bx$, the vector $\ba \in \Dd$
that maximizes the correlation also maximizes the likelihood
up to an affine change of contrast.

{\bfseries{Generalized Likelihood Ratio:}}
Motivated by optimality guarantees of the LR test
\eqref{eq:lr} and our previous work in \cite{deledalle2012compare},
a template matching criterion can be defined
from statistical detectors designed for
\textit{composite hypothesis problems}.
The generalized LR (GLR) replaces the unknowns
$\alpha$, $\beta$ and $\btheta$ in eq.~(\ref{eq:lr})
by their maximum likelihood estimates (MLE) under each hypothesis:
%the respective hypothesis $\mathcal{H}_0$ and $\mathcal{H}$:
\begin{align}\label{eq:glr}
  \mathcal{G}(\bx, \ba)
  &=
  \frac{
    \sup_{\alpha, \beta}
    p(\bx | \btheta = \transforad(\ba), \mathcal{H}_0)
  }{
    \sup_{\boldsymbol{t}}
    p(\bx | \btheta = \boldsymbol{t}, \mathcal{H}_1)
  }
  \nonumber\\
  & =
  \frac{
    p(\bx | \btheta = \transfonop_{\hat{\alpha}, \hat{\beta}}(\ba))
  }{
    p(\bx| \btheta = \hat{\boldsymbol{t}})
  }
\end{align}
where $\hat{\alpha}$, $\hat{\beta}$ and $\hat{\boldsymbol{t}}$ are the MLE of
the unknown $\alpha$, $\beta$ and $\btheta$.
By construction, the GLR satisfies the contrast invariance
property.
% Phrase pas tres claire pour moi:
Asymptotically to the SNR, GLR is optimal due to the efficiency of
MLE. Its asymptotic distribution
is known and so are the $P_{FA}$ values associated
to any given threshold $\tau$:
GLR is asymptotically a constant false alarm rate (CFAR) detector.
The GLR test is also invariant upon changes of variable \cite{kay2003invariance}:
it does not depend on the representation of the noisy patch.
While we noted that there are no UMP
detectors for our composite hypothesis problem, GLR is
\emph{asymptotically} UMP among invariant tests \cite{lehmann1959optimum}.
Due to its dependency on MLE,
the performance of GLR may fall in low SNR conditions, where
the MLE is known to behave poorly.

{\bfseries{Stabilization:}}
A classical approach to extend the applicability of a matching criterion
to non-Gaussian noises is to apply a transformation to the noisy
patches. The transformation is chosen so that the transformed patches
follow a (close to) Gaussian distribution with constant variance (hence
their name: variance-stabilization transforms).
This leads for instance to the homomorphic
approach which maps multiplicative noise to additive noise with stationary
variance. This is also the principle of Anscombe transform and its variants
used for Poisson noise.
%, but which unfortunately cannot stabilize the variance perfectly.
Given an application $s$ which stabilizes the variance
for a specific noise distribution, stabilization-based criteria
can be obtained using \eqref{eq:gauss_kernel} or \eqref{eq:glr} on the~output~of~$s$:
\begin{align}
  \mathcal{S}_\mathcal{C}(\bx, \ba)
  &= \mathcal{C}(s(\bx), s(\ba))~,\\
  \mathcal{S}_\mathcal{G}(\bx, \ba)
  &= \mathcal{G}(s(\bx), s(\ba))
\end{align}
where the likelihood function $p(s(\bx)| s(\btheta))$ is assumed
to be a Gaussian distribution centered on $s(\btheta)$ with a
covariance matrix $\sigma^2 \mathbf{I}$.
As we will see, an advantage of this approach compared to the GLR criterion is that
it is usually simpler to evaluate in closed-form, and then, leads to faster algorithms.
An important limitation of this approach lies nevertheless
in the existence of a stabilization function $s$.
Beyond existence, the performance of this approach may fall
if the transformed data distribution is far from the Gaussian distribution.

\section{GLR in different noise conditions}
\label{sec:examples}

In this section, we provide closed-form expressions or iterative schemes to
evaluate the GLR in the case of Gaussian noise, gamma noise
and Poisson noise.

\begin{prop}[Gaussian noise]
Consider that $\bX$ follows a Gaussian distribution such that
  \vspace{-0.1cm}
\begin{align*}
  p(x_k | \theta_k) = \frac{1}{\sqrt{2 \pi} \sigma}
  \exp \left(-\frac{ (x_k - \theta_k)^2}{2 \sigma^2} \right)~,
\end{align*}
and consider the class of affine transformations
$\transforad(\bx) = \alpha \bx + \beta \mathds{1}$.
In this case, we have
  \vspace{-0.1cm}
\begin{align*}
  -\log \mathcal{G}(\bx, \ba)
  &= (1 - \mathcal{C}(\bx, \ba)^2)
  \frac{
    \| {\bm x} - \bar{x} \mathds{1} \|_2^2
  }{
    2 \sigma^2
  }~.
\end{align*}
\end{prop}

\begin{proof}
  %Since $Var[X_k] = \sigma^2$ is independent of $\theta_k$, the variance
  %is stabilized with $s : x \mapsto x$.
  For the Gaussian law, the MLE of $\bm{\theta}$ is given by
  $\hat{\boldsymbol{t}} = \bx$ so that
  \vspace{-0.2cm}
  \begin{align*}
    -\log \mathcal{G}(\bx, \ba)
    = \frac{\| x - \hat{\alpha} \ba - \hat{\beta} \mathds{1} \|_2^2}{2 \sigma^2}
  \end{align*}
  and $\hat{\alpha}$ and $\hat{\beta}$
  are the coefficients of the linear least squared regression, i.e.,
  \vspace{-0.1cm}
  \begin{align*}
    \hat{\alpha} =
    \frac{
      \sum_{k=1}^N (x_k - \bar{x}) (a_k - \bar{a})
    }{
      \sum_{k=1}^N (a_k - \bar{a})^2
    }
    \quad \text{and} \quad
    \hat{\beta}
    =
    \bar{x} - \alpha \bar{a}~,
  \end{align*}
  with $\bar{x}$ and $\bar{a}$ the empirical mean of $\bx$ and $\ba$.
  Injecting the expression of $\hat{\alpha}$ and $\hat{\beta}$ in the previous equation
  gives the proposed formula.
\end{proof}

Remark that for a fixed observation $\bm x$ and
any $\bm a_1, \bm a_2 \in \Dd$, if
$\mathcal{C}(\bx, \ba_1) < \mathcal{C}(\bx, \ba_2)$ then
$\mathcal{G}(\bx, \ba_1) < \mathcal{G}(\bx, \ba_2)$.
In particular, we have
  \vspace{-0.1cm}
\begin{align*}
  \uargmax{\ba \in \Dd}
  \mathcal{G}(\bx, \ba)
  &=
  \uargmax{\ba \in \Dd}
  \mathcal{C}(\bx, \ba)\\
  &=
  \uargmax{\ba \in \Dd}
  \sup_{\alpha, \beta}
  p(\bx | \btheta = \transforad(\ba), \mathcal{H}_0)
\end{align*}
which is the MLE under the hypothesis $\mathcal{H}_0$.
However, beyond equivalence of their maxima,
the GLR is not equivalent to the correlation
even in the case of Gaussian noise.
They have different detection performance
when the purpose is to take a decision by thresholding
their answer.
Compared to the correlation, GLR adapts its answer with respect to
$\frac{
    \| {\bm x} - \bar{x} \mathds{1} \|^2
  }{
    2 \sigma^2
  }$
which, in some sense, measures the signal-to-noise-ratio (SNR) in $\bm x$.
For a fixed threshold $\tau$ and $\bm a$, if the SNR of $\bm x$ is small enough, GLR will
put the pair $(\bm x, \bm a)$ in correspondence whatever their content.
In fact, when the SNR is small enough, any template up to a radiometric transform
can explain the observed realization. The correlation, which does not take into
account the noise in its definition, does not adapt to the SNR of $\bm x$.
Worse, the correlation tends to increase when the SNR of $\bm x$ decreases.
We will see in Section \ref{sec:examples} that such a behavior of GLR is of main
importance for a template matching task.

\begin{prop}[Gamma noise]
Consider that $\bX$ follows a gamma distribution such that
  \vspace{-0.1cm}
\begin{align*}
  p(x_k | \theta_k) = \frac{L^L x_k^{L-1}}{\Gamma(L) \theta_k^L}
  \exp\left( -\frac{L x_k}{\theta_k}\right)
\end{align*}
and consider the class of log-affine transformations
$\transforad(\bx) = \beta \bx^\alpha$
where $(.)^\alpha$ is the element-wise power function.
 In this case, we have
  \vspace{-0.1cm}
\begin{align*}
  -\log \mathcal{G}(\bx, \ba)
  &=
  L \sum_{k =1}^N \log\left( \frac{\hat{\beta} a_k^{\hat{\alpha}}}{x_k} \right)
\end{align*}
where $\hat{\alpha}$ and $\hat{\beta}$ can be obtained iteratively as
  \vspace{-0.1cm}
\begin{align*}
  %\left\{
      \hat{\alpha}_{i+1} \!=\! \hat{\alpha}_i -
      \frac{
        \sum_{k} \!\left( 1 - r_{k,i}\right) \log a_k
      }{
        \sum_{k} r_{k,i} (\log a_k)^2
      }
      \quad \text{and} \quad
      \displaystyle
      \hat{\beta}_{i+1} \!=\! \frac{1}{N} \sum_{k} \frac{x_k}{a_k^{\hat{\alpha}_i}}
  %\right.
\end{align*}
with $r_{k,i} = x_k / (\hat{\beta}_i a_k^{\hat{\alpha}_i})$,
whatever the initialization.
\end{prop}

\begin{proof}
  %Since $Var[\log X_k] = \psi(1, L)$ is independent of $\theta_k$, the variance
  %is stabilized with $s : x \mapsto \log x$.
  For the gamma law, the MLE of $\bm{\theta}$ is given by
  $\hat{\boldsymbol{t}} = \bx$ so that
  \vspace{-0.2cm}
  \begin{align*}
    -\log \mathcal{G}(\bx, \ba)
    = L \sum_{k=1}^N \left(
      \log \frac{\hat{\beta} a_k^{\hat{\alpha}}}{x_k}
      +
      \frac{x_k}{\hat{\beta} a_k^{\hat{\alpha}}}
      -
      1
    \right)~.
  \end{align*}
  The function $\beta \mapsto \sum_k -\log p(x_k | \beta a_k^\alpha)$ has a unique minimum
  at $\frac{1}{N} \sum_{k} \frac{x_k}{a_k^{\alpha}}$. Moreover,
  the function $\alpha \mapsto \sum_k -\log p(x_k | \theta_k = \beta a_k^\alpha)$ is convex and
  twice differentiable, therefore the Newton method can be used to estimate $\hat{\alpha}$
  whatever the initialization.
  Differentiating twice $\alpha \mapsto \sum_k -\log p(x_k | \theta_k = \beta a_k^\alpha)$
  gives the proposed iterative scheme.
  Injecting the value of $\hat{\alpha}$ and $\hat{\beta}$ in the previous equation
  gives the proposed formula.
\end{proof}

Unlike in the case of the Gaussian law, there is no closed-form formula of GLR in the
case of the gamma law and one should rather compute it iteratively.
Note that in practice only a few iterations are required
if one initializes using the log-moment estimation, as suggested in \cite{nicolas2002introduction},
leading to the following initialization:
  \vspace{-0.1cm}
\begin{align*}
  \hat{\alpha}_0 &= \sqrt{\frac{
      \max(\sum_k (\log x_k - \overline{\log x})^2 - \psi(1,L), 0)
    }{
      \sum_k (\log a_k - \overline{\log a})^2
    }}\\
  \hat{\beta}_0 &=
  \exp \left( \overline{\log x} - \psi(L) + \log(L) - \alpha \overline{\log a} \right)
\end{align*}
where $\overline{\log x} = \frac{1}{N} \sum_k \log x_k$ and
$\overline{\log a} = \frac{1}{N} \sum_k \log a_k$.

\begin{figure*}[!t]
  \centering
  \subfigure[]{\includegraphics[width=0.24\linewidth]{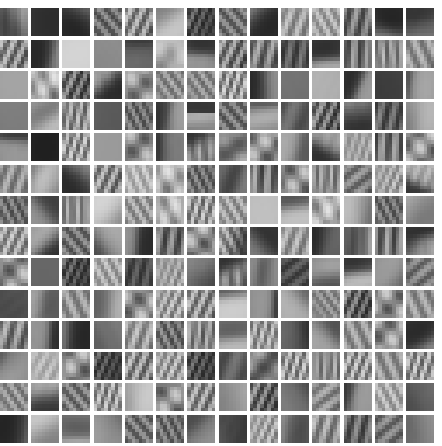}}
  \subfigure[]{\includegraphics[width=0.24\linewidth]{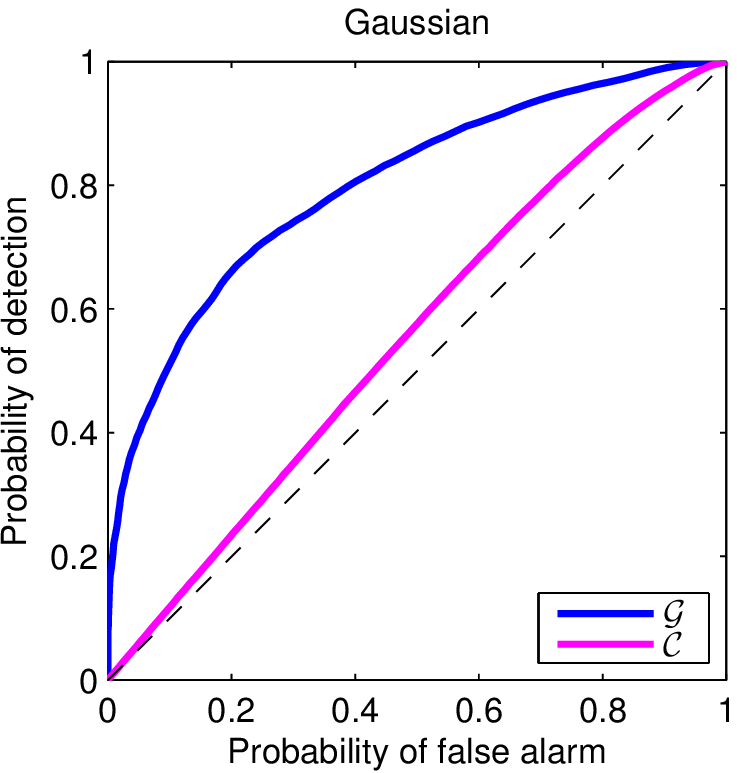}}
  \subfigure[]{\includegraphics[width=0.24\linewidth]{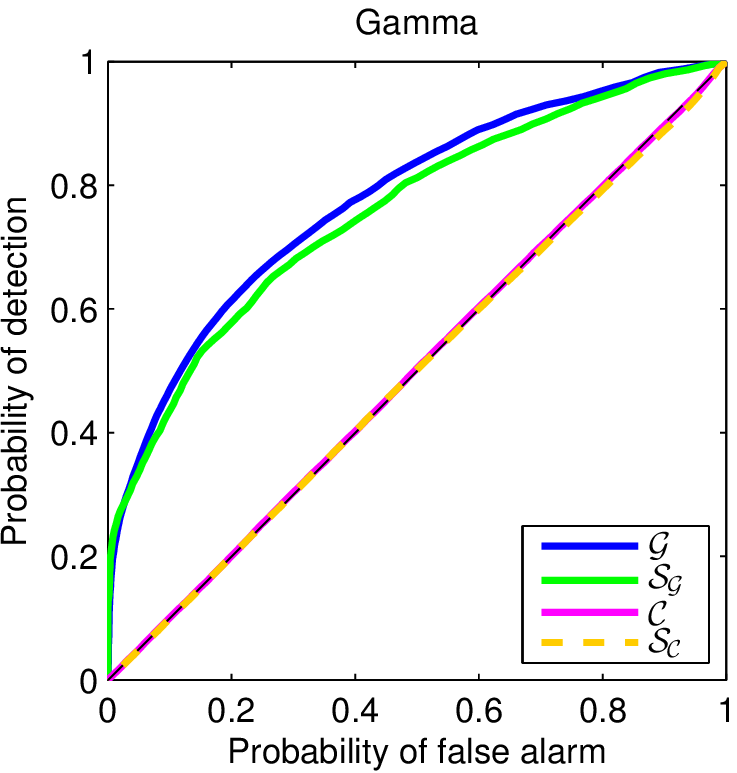}}
  \subfigure[]{\includegraphics[width=0.24\linewidth]{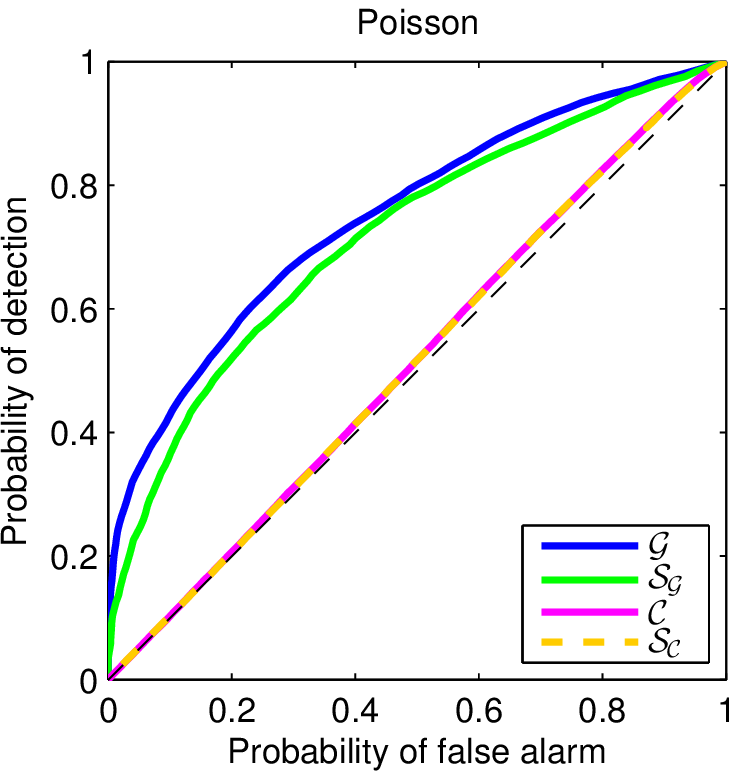}}
  \caption{(a) Patch dictionary.
    (b) ROC curve obtained under Gaussian noise,
    (c) ROC curve obtained under gamma noise
    and (d) ROC curve obtained under Poisson noise.
    In all experiments, the SNR over the whole dictionary is about $-3$dB.
  }
  \label{fig:roc}
\end{figure*}

\begin{prop}[Poisson noise]
Consider that $\bX$ follows a Poisson distribution so that
  \vspace{-0.1cm}
\begin{align*}
  p(x_k | \theta_k) = \frac{\theta_k^{x_k} e^{-\theta_k}}{x_k!}
\end{align*}
and consider the class of log-affine transformations
$\transforad(\bx) = \beta \bx^\alpha$.
In this case, we have
  \vspace{-0.1cm}
\begin{align*}
  -\log \mathcal{G}(\bx, \ba)
  &=
  \sum_{k =1}^N x_k \log\left( \frac{x_k}{\hat{\beta} a_k^{\hat{\alpha}}} \right)
\end{align*}
where $\hat{\alpha}$ and $\hat{\beta}$ can be obtained iteratively as
  \vspace{-0.1cm}
\begin{align*}
  \hat{\alpha}_{i+1} \!=\! \hat{\alpha}_i -
  \frac{
    \sum_{k} \!\left( \hat{\beta}_i a_k^{\hat{\alpha}_i} - x_k\right) \!\log a_k
  }{
    \sum_{k} \hat{\beta}_i a_k^{\hat{\alpha}_i} (\log a_k)^2
  }
  \quad \!\text{and}\! \quad
  \displaystyle
  \hat{\beta}_{i+1} \!=\! \frac{\sum_{k} x_k}{\sum_{k} a_k^{\hat{\alpha}_i}}
\end{align*}
whatever the initialization.
\end{prop}

\begin{proof}
  %Since, for $\theta_k$ big enough, $Var[2 (X_k + 3/8)^{1/2}] \approx 1$
  %is independent of $\theta_k$, the variance
  %is stabilized with $s : x \mapsto (s + 3/8 \mathds{1} )^{1/2}$.
  For the Poisson law, the MLE of $\bm{\theta}$ is given by
  $\hat{\boldsymbol{t}} = \bx$ such that
  \vspace{-0.1cm}
  \begin{align*}
    -\log \mathcal{G}(\bx, \ba)
    =
    \sum_{k =1}^N \left(
      x_k \log\left( \frac{x_k}{\hat{\beta} a_k^{\hat{\alpha}}} \right)
      + \hat{\beta} a_k^{\hat{\alpha}}
      - x_k
    \right)~.
  \end{align*}
  The function $\beta \mapsto \sum_k -\log p(x_k | \beta a_k^\alpha)$ has a unique minimum
  at $\frac{\sum_{k} x_k}{\sum_{k} a_k^{\alpha}}$. Moreover,
  the function $\alpha \mapsto \sum_k -\log p(x_k | \theta_k = \beta a_k^\alpha)$ is convex and
  twice differentiable, such that the Newton method can be used to estimate $\hat{\alpha}$
  whatever the initialization.
  Differentiating twice $\alpha \mapsto \sum_k -\log p(x_k | \theta_k = \beta a_k^\alpha)$
  gives the proposed iterative scheme.
  Injecting the value of $\hat{\alpha}$ and $\hat{\beta}$ in the previous equation
  gives the proposed formula.
\end{proof}

Again there is no closed-form formula of GLR, but
in practice only a few iterations are required
if one uses the $\hat{\alpha}$ and $\hat{\beta}$ that minimize
the linear least square error between $\log {\bx}$ and $\log {\ba}$.

\section{Evaluation of performance}
\label{sec:evaluation}

\subsection{Detection performance}

We evaluate the relative performance of the correlation, GLR and
the variance stabilization based matching criteria on a dictionary
composed of $196$ noise-free patches
of size $N\!=\!8\!\times\!8$. The noise-free
patches have been obtained using the k-means on patches extracted from the classical
$512\!\times\!512$ \textit{Barbara} image.
Each noisy patch $\bx$ is a noisy realization of the noise-free patches
under Gaussian, gamma or Poisson noise with an overall SNR of about $-3$dB.
Each template $\ba$ is a randomly transformed atom of the dictionary
up to an affine change of contrast for the experiments involving Gaussian noise,
and up to a log-affine change of contrast
under gamma or Poisson noises.
All criteria are evaluated for all pairs $(\bx, \ba)$.
The process is repeated $20$ times with independent noise realizations
and radiometric transformations.

The performance of the matching criteria is given in term of their
\textit{receiver operating characteristic} (ROC) curve, i.e., the curve
of $P_D$ with respect to
$P_{FA}$, where we have relaxed the hypothesis test as
  \vspace{-0.1cm}
\begin{equation*}
  \begin{array}{llr}
    \mathcal{H}_0 : \exists \alpha, \beta \quad &\boldsymbol\theta \approx \transforad(\ba)
    & \text{(null hypothesis)} ,\\
    \mathcal{H}_1 :
    \forall \alpha, \beta \quad
    & \boldsymbol\theta \not\approx \transforad(\ba)
    & \text{(alternative hypothesis)}
  \end{array}
\end{equation*}
and where $\boldsymbol\theta \approx \transforad(\ba)$ reads as: on average, the noise-free patch
$\transforad(\ba)$ explains almost as well the realizations of
$\bX$ than the actual noise-free patch $\btheta$, and is measured by:
  \vspace{-0.1cm}
\begin{align*}
  \Dd_{KL}( \btheta \;\|\; \transforad(\ba) ) \leq \nu\,,
\end{align*}
where $\Dd_{KL}$ is the Kullback-Leibler divergence and $\nu$ is a
small value (chosen here equal to $0.02$).
Results are given in Figure \ref{fig:roc}.
Even with Gaussian noise or with variance stabilization,
the correlation behaves poorly in noisy condition.
The generalized likelihood ratio (GLR) is
the most powerful criterion followed by
the GLR with variance stabilization.

\begin{figure*}[!t]
  \centering
  % \begin{sideways}Gamma\end{sideways}
  \subfigure[]{\includegraphics[width=0.24\linewidth,viewport=0 0 128 112,clip]{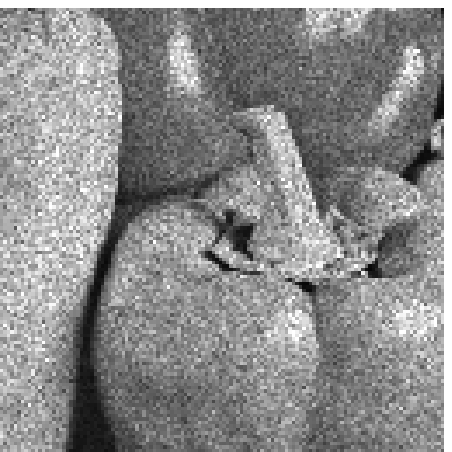}}
  \hfill
  \subfigure[]{\includegraphics[width=0.24\linewidth,viewport=0 0 128 112,clip]{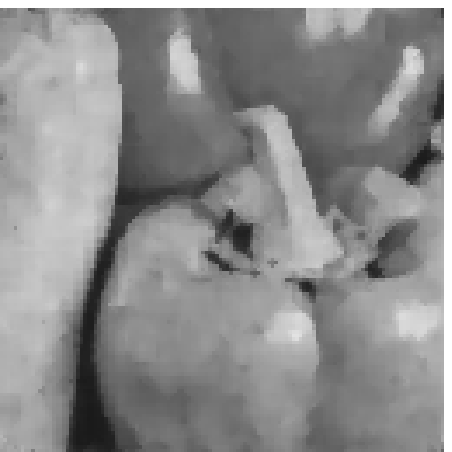}}
  \hfill
  \subfigure[]{\includegraphics[width=0.24\linewidth,viewport=0 0 128 112,clip]{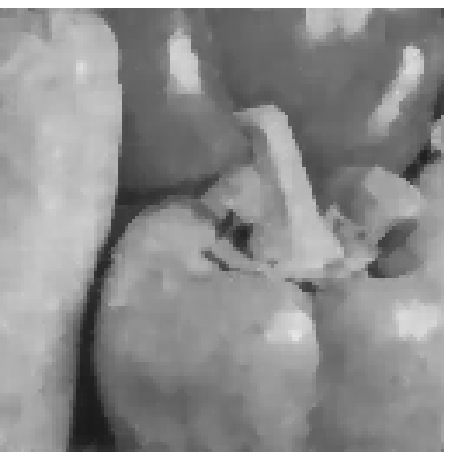}}
  \hfill
  \subfigure[]{\includegraphics[width=0.24\linewidth,viewport=0 0 128 112,clip]{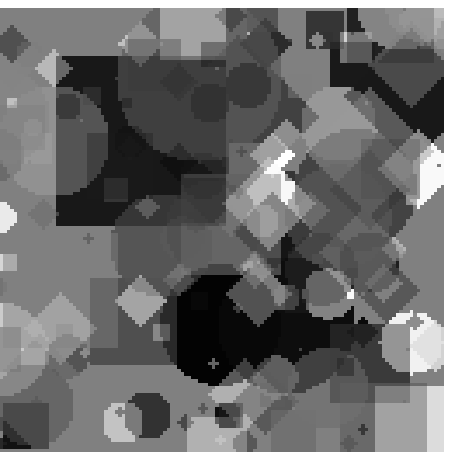}}\\[-0.2cm]
  \caption{(a) Noisy input image damaged by gamma noise (PSNR=21.14).
    (b) Denoised image using the GLR after variance stabilization
    (PSNR=27.42).
    (c) Denoised image using the GLR adapted to gamma noise
    (PSNR=27.53).
    (d) Image composed of the atoms of the dictionary.
  }
  \label{fig:denoising_results}
\end{figure*}

\subsection{Application to dictionary-based denoising}

We exemplify here the performance of GLR in a dictionary-based denoising task.
The dictionary $\Dd$ is considered
describing a generative model of the patches $\bx$ of the noisy image as
realizations of $\bX$ following a distribution of parameter $\btheta = \transforad(\ba)$
with $\ba \in \Dd$.
Under this model, we suggest estimating each patch of the image as:
  \vspace{-0.1cm}
\begin{align}\label{eq:glr_denoising}
  \hat{\btheta}(\bx)
  =
  \frac{1}{Z} \sum_{\ba \in \Dd} \mathcal{G}(\bx , \ba) \ba^\star
  \quad
  \text{with}
  \quad
  Z = \sum_{\ba \in \Dd} \mathcal{G}(\bx, \ba)\,,
\end{align}
\vspace*{-.48\baselineskip}\hfill\\
where $\ba^\star = \transfonop_{\hat{\alpha},\hat{\beta}}(\ba)$
and $\hat{\alpha}$ and $\hat{\beta}$ are the MLE of $\alpha$ and $\beta$
used in the calculation of $\mathcal{G}(\bx, \ba)$.
Equation \eqref{eq:glr_denoising} has a Bayesian interpretation as
the posterior mean estimator:
\vspace*{-1.1\baselineskip}\hfill\\
\begin{align}
  \hat{\btheta}(\bx)
  =
  \frac{\sum_{\ba \in \Dd} p(\ba^\star | \bx) \ba^\star}{\sum_{\ba \in \Dd} p(\ba^\star | \bx)}\,,
\end{align}
\vspace*{-.75\baselineskip}\hfill\\
considering a priori that the frequencies of the atoms of $\Dd$ are uniform in the image.
The posterior mean is known to minimize the Bayesian least square error
$\mathbb{E} \left[ \| \hat{\btheta}(\bX) - \btheta \|_2^2 \;|\; \btheta \right]$.

Figure \ref{fig:denoising_results} shows the denoising results obtained on a
$128 \times 128$ image damaged by gamma noise (with $L = 10$) using
\eqref{eq:glr_denoising} with the GLR adapted to gamma noise and
with the GLR adapted to a Gaussian law after variance stabilization%
\footnote{when using stabilization, a debiasing step is performed following \cite{xie2002ssr}.}.
The dictionary $\Dd$
is chosen as the set of all atoms extracted from a $128 \times 128$
image (a.k.a., an epitome) built following the transparent dead leaves model
of \cite{galerne2012transparent}.
This model ensures the dictionary to be shift invariant \cite{jost2006motif,benoit2011sparse}
while representing information of different scales.
As in \cite{jost2006motif,benoit2011sparse}, we manipulate epitomes
in Fourier domain in order to evaluate eq.~\eqref{eq:glr_denoising} efficiently.
%In the case of the gamma law, this can be done only
%if the values $\alpha$ are limited and fixed in advance.
%Here, we have chosen $\alpha \in [0.1, 0.5, 1, 1.5]$.
Eventually, Fig.~\ref{fig:denoising_results} shows that using the GLR
for the gamma law or for the Gaussian law after stabilizing
the variance are both satisfactory visually and in term of PSNR.

\section{Conclusion}
Normalized correlation is widely used as a contrast-invariant
criterion for template matching. We have shown that the GLR test
provides a criterion that is more robust to noise. In the case of
Gaussian noise, this criterion involves both a normalized correlation
term and a term that evaluates the signal-to-noise ratio of the noisy data.
Under non-Gaussian noise distributions, criteria derived
from the GLR test are generally not known in closed form but require a
few iterations to be evaluated. When variance stabilization technique
can be employed, our numerical experiments show that good performance
is reached using Gaussian GLR after variance stabilization.

% References should be produced using the bibtex program from suitable
% BiBTeX files (here: strings, refs, manuals). The IEEEbib.bst bibliography
% style file from IEEE produces unsorted bibliography list.
% -------------------------------------------------------------------------
\small
\bibliographystyle{IEEEbib}

%\bibliography{biblio}

\end{document}